\documentclass[11pt,journal,twoside,onecolumn,draftcls]{IEEEtran}

\usepackage{graphicx,subfigure}
\usepackage{amsfonts,amsmath,amssymb}
\usepackage{epic,eepic,eepicemu}
\usepackage{epsf}
\usepackage{epsfig}
\usepackage{graphics}
\usepackage{psfrag}
\newtheorem{theorem}{Theorem}
\newtheorem{lemma}{Lemma}

\newtheorem{definition}{Definition}

\makeatletter
\def\@cite#1#2{[\if@tempswa #2 \fi #1]}
\makeatother

\newcommand{\beq}{\begin{equation}}
\newcommand{\eeq}{\end{equation}}
\newcommand{\bea}{\begin{eqnarray}}
\newcommand{\eea}{\end{eqnarray}}

\newtheorem{discussion}{Discussion}

\newtheorem{corollary}{Corollary}
\newtheorem{remark}{Remark}
\long\def\symbolfootnote[#1]#2{\begingroup
\def\thefootnote{\fnsymbol{footnote}}\footnote[#1]{#2}\endgroup}

\long\def\comment#1{}

\begin{document}

\title{\huge{A Generalized Cut-Set Bound}}
%
\author{Amin Aminzadeh Gohari and Venkat Anantharam  \\
Department of Electrical Engineering and Computer Science \\ University of California, Berkeley\\
\texttt{\small $\{$aminzade,ananth$\}$@eecs.berkeley.edu} \\
} \maketitle
\begin{abstract}
In this paper, we generalize the well known cut-set bound to the
problem of lossy transmission of functions of arbitrarily correlated
sources over a discrete memoryless multiterminal network.

\end{abstract}
\section{Introduction}
A general multiterminal network is a model for reliable
communication of sets of messages among the nodes of a network, and
has been extensively used in modeling of wireless systems. It is
known that unlike the point-to-point scenario, in a network scenario
the separation of the source and channel codings is not necessarily
optimal \cite{CoverGamalSalehi}. In this paper we study the
limitations of joint source-channel coding strategies for lossy
transmission across multiterminal networks.

A discrete memoryless general multiterminal network (GMN) is
characterized by the conditional distribution
$$q(y^{(1)},y^{(2)},...,y^{(m)}|x^{(1)},x^{(2)},...,x^{(m)}),$$ where
$X^{(i)}$ and $Y^{(i)}$ ($1\leq i \leq m$) are respectively the
input and the output of the channel at the $i^{th}$ party. In a
general multiterminal channel with correlated sources, the $m$ nodes
are observing i.i.d. repetitions of $m$, possibly correlated, random
variables $W^{(i)}$ for $1\leq i \leq m$. The $i^{th}$ party ($1\leq
i \leq m$) has access to the i.i.d. repetitions of $W^{(i)}$, and
wants to reconstruct, within a given distortion, the i.i.d.
repetitions of a function of all the observations, i.e.
$f^{(i)}(W^{(1)}, W^{(2)}, ..., W^{(m)})$ for some function
$f^{(i)}(\cdot)$. If this is asymptotically possible within a given
distortion (see section \ref{Section:Definitions} for a formal
definition), we call the source $(W^{(1)}, W^{(2)}, ..., W^{(m)})$
admissible. In some applications, each party may be interested in
recovering i.i.d. repetitions of functions of the observations made
at different nodes. In this case the function $f^{(i)}(W^{(1)},
W^{(2)}, ..., W^{(m)})$ takes the special form of
$\big(f^{(i,1)}(W^{(1)}), f^{(i,2)}(W^{(2)}), ...,
f^{(i,m)}(W^{(m)})\big)$ for some functions $f^{(i,j)}(\cdot)$.

\begin{figure}
\centering
\includegraphics[width=100mm]{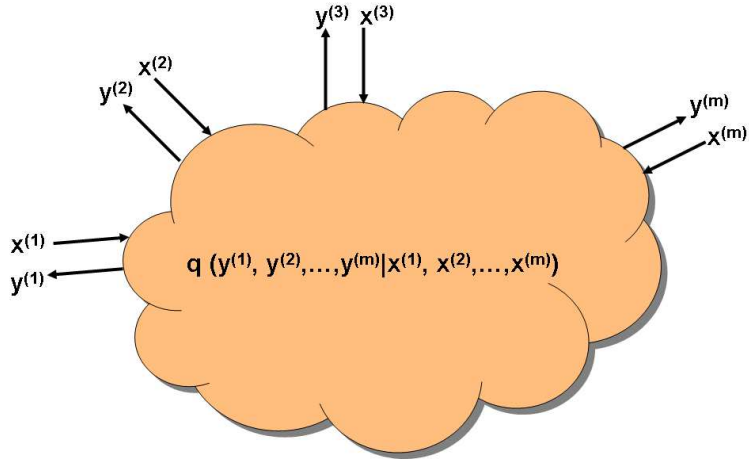}
\caption{The statistical description of a
network.}\label{fig:networkpicture}
\end{figure}

The admissible source region of a general multiterminal network is
not known when the sources are independent except in certain special
cases; less is known when the sources are allowed to be arbitrarily
correlated. It is known that the source$-$channel separation theorem
in a network scenario breaks down \cite{CoverGamalSalehi}. In this
paper, we prove a new outer bound on the admissible source region of
GMNs. Specializing by requiring zero distortion at the receivers,
assuming that the functions $f^{(i)}(W^{(1)}, W^{(2)}, ...,
W^{(m)})$ ($1\leq i \leq m$) have the form of $(f^{(i,1)}(W^{(1)}),
f^{(i,2)}(W^{(2)}), ..., f^{(i,m)}(W^{(m)}))$, and that the
individual messages $f^{(i,j)}(W^{(j)})$ are mutually independent,
our result reduces to the well known cut-set bound. The results can
be carried over to the problem of ``lossless transmission" for the
following reason: requiring the $i^{th}$ party to reconstruct the
i.i.d. repetitions of $f^{(i)}(W^{(1)}, W^{(2)}, ..., W^{(m)})$ with
arbitrarily small average probability of error is no stronger than
requiring the $i^{th}$ party to reconstruct the i.i.d repetitions of
$f^{(i)}(W^{(1)}, W^{(2)}, ..., W^{(m)})$ with a vanishing average
distortion (for details see section \ref{Section:Definitions}).
Other extensions of cut-set bound can be found in \cite{Gastpar} and
\cite{KramerSavari}. Furthermore some existing works show the
possibility and benefit of function computation during the
communication (see for instance
\cite{NazerGastpar}\cite{GiridharKumar}\cite{OrlitskyRoche}\cite{Yamamoto}\cite{AppuswamyFranceschettiKaramchandaniZeger}).

A main contribution of this paper is its proof technique which is
based on the ``potential function method" introduced in \cite{GA_SM}
and \cite{GA_CM}. Instead of taking an arbitrary network and proving
the desired outer bound while keeping the network fixed throughout,
we consider a function from the set of all $m$-input/$m$-output
discrete memoryless networks to subsets of $\mathbb{R}_{+}^c$, where
$\mathbb{R}_{+}^c$ is the set of all $c$-tuples of non-negative
reals. We then identify properties of such a function which would
need to be satisfied in one step of the communication for it to give
rise to an outer bound. The generalized cut-set bound is then proved
by a verification argument. Properties that such a function would
need to satisfy are identified, intuitively speaking, as follows:
take an arbitrary code of length say $n$ over a multiterminal
network. During the simulation of the code, the information of the
parties begins from the $i^{th}$ party having the i.i.d. repetitions
of the random variable $W^{(i)}$; gradually evolves over time with
the usage of the network; and eventually after $n$ stages of
communication reaches its final state where the parties know enough
to estimate their objectives within the desired average distortion.
The idea is to quantify this gradual evolution of information;
\emph{bound the derivative of the information growth at each stage}
from above by showing that one step of communication can buy us at
most a certain amount; and conclude that at the final stage, i.e.
the $n^{th}$ stage, the system can not reach an information state
better than $n$ times the outer bound on the derivative of
information growth. An implementation of this idea requires
quantification of the information of the $m$ parties at a given
stage of the process. To that end, we evaluate the function we
started with at a \emph{virtual channel} whose inputs and outputs
represent, roughly speaking, the initial and the gained knowledge of
the parties at the given stage of the communication. See Lemma
\ref{Thm:MainThm} of section \ref{Section:StatementOfResults} and
the proof of Theorem \ref{Thm:CutSet} of section
\ref{Section:Proofs} for a formal formulation.

The outline of this paper is as follows. In section
\ref{Section:Definitions}, we introduce the basic notations and
definitions used in this paper. Section
\ref{Section:StatementOfResults} contains the main results of this
paper followed by section \ref{Section:Proofs} which gives formal
proofs for the results. Appendices \ref{Appendix:CutSet} and
\ref{Appendix:ExtensionFive} complete the proof of Theorem
\ref{Thm:CutSet} from section \ref{Section:CutSet}.

\section{Definitions and Notation} \label{Section:Definitions}
\begin{table}
\caption{Notations} \centering
\begin{tabular}{|c|c|}
    \hline
Variable &  Description  \\
\hline
    \footnotesize{$\mathbb{R}$}& \footnotesize{Real numbers.} \\
\hline
\footnotesize{$\mathbb{R}_{+}$} & \footnotesize{Non-negative real numbers.}\\
\hline
\footnotesize{$[k]$} & \footnotesize{The set $\{1,2,3,...,k\}.$}\\
\hline
\footnotesize{$m$} & \footnotesize{Number of nodes of the network.}\\
\hline
\footnotesize{$q(y^{(1)},...,y^{(m)}|x^{(1)},...,x^{(m)})$}& \footnotesize{The statistical description of a}\\& \footnotesize{ multi-terminal network.} \\
\hline
\footnotesize{$W^{(i)}$} & \footnotesize{Random variable representing the source observed at the $i^{th}$ node.}\\
\hline
\footnotesize{$M^{(i)}$}& \footnotesize{Random variable to be reconstructed,}\\&\footnotesize{\ in a possibly lossy way, at the $i^{th}$ node.} \\
\hline \footnotesize{$\mathcal{X}^{(i)}, \mathcal{Y}^{(i)}$,
$\mathcal{W}^{(i)}$,
$\mathcal{M}^{(i)}$}& \footnotesize{Alphabet sets of $X^{(i)}$, $Y^{(i)}$, $W^{(i)}$, $M^{(i)}$.} \\
\hline
\footnotesize{$\Delta^{(i)}(\cdot, \cdot)$}& \footnotesize{Distortion function used by the $i^{th}$ party. }\\
\hline \footnotesize{$\zeta_k^{(i)}(\cdot)$}&\footnotesize{The
encoding function used by the $i^{th}$ party at the $k^{th}$ stage.}\\
\footnotesize{$\vartheta^{(i)}(\cdot)$}&\footnotesize{The decoding
function
at the $i^{th}$ party.}\\
\hline
    \footnotesize{$n$}&\footnotesize{Length of the code used.}\\
\hline
    \footnotesize{$\Pi(\cdot)$}& \footnotesize{Down-set (Definition \ref{def:downset});}\\
    \footnotesize{$\oplus$}&\footnotesize{Minkowski sum of two sets (Definition \ref{def:oplus}).}\\
    \footnotesize{$\geq$}&\footnotesize{A vector or a set being greater than or equal the other (Definition \ref{def:downset}).}\\
\hline \footnotesize{$\Psi$}&\footnotesize{A permissible set of
input distributions;} \\ &\footnotesize{Given input sources and a
multiterminal network, $\Psi$ is a set of }\\&\footnotesize{joint
distributions on $\mathcal{X}^{(1)} \times \mathcal{X}^{(2)} \times
\mathcal{X}^{(3)} \times
\cdot\cdot\cdot\times \mathcal{X}^{(m)}$.} \\
&\footnotesize{Inputs to the
 network have a joint distribution belonging to this set}.\\
 \hline
    \end{tabular}
\end{table}
Throughout this paper we assume that each random variable takes
values in a finite set. $\mathbb{R}$ denotes the set of real numbers
and $\mathbb{R}_{+}$ denotes the set of non-negative reals. For any
natural number $k$, let $[k]=\{1,2,3,...,k\}$. For a set $S\subset
[k]$, let $S^c$ denote its compliment, that is $[k]-S$. The context
will make the ambient space of $S$ clear.

We represent a GMN by the conditional distribution
$$q(y^{(1)},y^{(2)},...,y^{(m)}|x^{(1)},x^{(2)},...,x^{(m)})$$
meaning that the input by the $i^{th}$ party is $X^{(i)}$ and the
output at the $i^{th}$ party is $Y^{(i)}$. We assume that the
$i^{th}$ party ($1\leq i \leq m$) has access to i.i.d. repetitions
of $W^{(i)}$. The message that needs to be delivered (in a possibly
lossy manner) to the $i^{th}$ party is taken to be
$M^{(i)}=f^{(i)}(W^{(1)}, W^{(2)}, ..., W^{(m)})$ for some function
$f^{(i)}(\cdot)$. We assume that for any $i\in [m]$, random
variables $X^{(i)}$, $Y^{(i)}$, $W^{(i)}$ and $M^{(i)}$ take values
from discrete sets $\mathcal{X}^{(i)}$, $\mathcal{Y}^{(i)}$,
$\mathcal{W}^{(i)}$ and $\mathcal{M}^{(i)}$ respectively. For any
natural number $n$, let $(\mathcal{X}^{(i)})^n$,
$(\mathcal{Y}^{(i)})^n$, $(\mathcal{W}^{(i)})^n$ and
$(\mathcal{M}^{(i)})^n$ denote the $n$-th product sets of
$\mathcal{X}^{(i)}$, $\mathcal{Y}^{(i)}$, $\mathcal{W}^{(i)}$ and
$\mathcal{M}^{(i)}$. We use $Y_{1:k}^{(i)}$ to denote
$(Y_{1}^{(i)},Y_{2}^{(i)},...,Y_{k}^{(i)})$.

For any $i \in [m]$, let the distortion function $\Delta^{(i)}$ be a
function $\Delta^{(i)}:\mathcal{M}^{(i)}\times \mathcal{M}^{(i)}
\rightarrow [0,\infty)$ satisfying $\Delta^{(i)}(m^{(i)},m^{(i)})=0$
for all $m^{(i)}\in \mathcal{M}^{(i)}$. For any natural number $n$
and vectors $(m_1^{(i)}, m_2^{(i)},...,m_n^{(i)})$ and $(m_1^{'(i)},
m_2^{'(i)}, ...,m_n^{'(i)})$ from $(\mathcal{M}^{(i)})^n$, let
$$\Delta^{(i)}_{n}(m_{1:n}^{(i)},m_{1:n}^{'(i)})=\frac{1}{n}\sum_{k=1}^n\Delta^{(i)}(m_{k}^{(i)}, m_{k}^{'(i)}).$$
Roughly speaking, we require the i.i.d. repetitions of random
variable $M^{(i)}$ to be reconstructed, by the $i^{th}$ party,
within the average distortion of $D^{(i)}$.

\begin{definition} Given natural number $n$, an $\emph{(n)}$\emph{-code} is
the following set of mappings:
\begin{eqnarray*}
\mbox{For any }i\in [m]:&
\zeta_1^{(i)}:&(\mathcal{W}^{(i)})^n\longrightarrow \mathcal{X}^{(i)};\\
\mbox{For any } i\in [m],k\in [n]-\{1\}:&
\zeta_k^{(i)}:&(\mathcal{W}^{(i)})^n \times
(\mathcal{Y}^{(i)})^{k-1}
\longrightarrow \mathcal{X}^{(i)};\\
\mbox{For any }i \in [m]:&\vartheta^{(i)}:&(\mathcal{W}^{(i)})^n
\times (\mathcal{Y}^{(i)})^{n} \longrightarrow
(\mathcal{M}^{(i)})^n.
\end{eqnarray*}
Intuitively speaking $\zeta_k^{(i)}$ is the encoding function of the
$i^{th}$ party at the $k^{th}$ time instance, and $\vartheta^{(i)}$
is the decoding function of the $i^{th}$ party.

Given positive reals $\epsilon$ and $D^{(i)}$ ($1\leq i \leq m$),
and a source marginal distribution $p(w^{(1)},
w^{(2)},...,w^{(m)})$, an $(n)$-code is said to satisfy the average
distortion interval $D^{(i)}$ (for all $i \in [m]$) over the channel
$q(y^{(1)},y^{(2)},...,y^{(m)}|x^{(1)},x^{(2)},...,x^{(m)})$ if the
following ``average distortion" condition is satisfied:

Assume that random variables $W^{(i)}_{1:n}$ for $i\in [m]$ are $n$
i.i.d. repetition of random variables $(W^{(1)}, W^{(2)},... ,
W^{(m)})$ with joint distribution $p(w^{(1)}, w^{(2)},...,w^{(m)})$.
Random variables $X_k^{(i)}$ and $Y_k^{(i)}$ ($k\in [n]$, $i\in[m]$)
are defined according to the following constraints:
$$p(w_{1:n}^{(1)},
w_{1:n}^{(2)}, ..., w_{1:n}^{(m)}, x_{1:n}^{(1)},
x_{1:n}^{(2)},...,x_{1:n}^{(m)}, y_{1:n}^{(1)},
y_{1:n}^{(2)},...,y_{1:n}^{(m)})=$$$$\prod_{k=1}^n p(w_{k}^{(1)},
w_{k}^{(2)}, ..., w_{k}^{(m)})\times \prod_{k=1}^n q(y_k^{(1)},
y_k^{(2)},...,y_k^{(m)}|x_k^{(1)}, x_k^{(2)},...,x_k^{(m)})\times
\prod_{k=1}^n \prod_{i=1}^m p(x_k^{(i)}|w_{1:n}^{(i)},
y_{1:k-1}^{(i)});$$ and that
$X_1^{(i)}=\zeta_1^{(i)}\big(W_{1:n}^{(i)}\big),$ and for any $2\leq
k \leq n$, $X_k^{(i)}=\zeta_k^{(i)}\big(W_{1:n}^{(i)},
Y_{1:k-1}^{(i)}\big).$ Random variables $X_k^{(i)}$ and $Y_k^{(i)}$
are representing the input and outputs of the $i^{th}$ party at the
$k^{th}$ time instance and satisfy the following Markov chains:
$$W_{1:n}^{(1)}...W_{1:n}^{(m)}
Y_{1:k-1}^{(1)}...Y_{1:k-1}^{(m)}-W_{1:n}^{(i)}Y_{1:k-1}^{(i)}-X_k^{(i)},$$
$$W_{1:n}^{(1)}...W_{1:n}^{(m)}
Y_{1:k-1}^{(1)}...Y_{1:k-1}^{(m)}-X_k^{(1)}...X_k^{(m)}-Y_k^{(1)}...Y_k^{(m)}.$$
We then have the following constraint for any $i \in [m]$:
$$\mathbb{E}\bigg[ \Delta_n^{(i)}\bigg(\vartheta^{(i)}\big(W_{1:n}^{(i)}, Y_{1:n}^{(i)}\big), M_{1:n}^{(i)}\bigg)\bigg]\leq D^{(i)}+\epsilon,$$ where
$M_k^{(i)}=f^{(i)}(W_k^{(1)}, W_k^{(2)}, ..., W_k^{(m)})$.
\end{definition}

\begin{definition} \label{def:capacityregion} Given positive reals $D^{(i)}$, a source marginal distribution
$p(w^{(1)}, w^{(2)},...,w^{(m)})$ is called an \emph{admissible
source} over the channel
$q(y^{(1)},y^{(2)},...,y^{(m)}|x^{(1)},x^{(2)},...,x^{(m)})$ if for
every positive $\epsilon$ and sufficiently large $n$, an $(n)$-code
satisfying the average distortion $D^{(i)}$, exists.

The ``independent messages zero distortion capacity region" of the
GMN,
$$C(q(y^{(1)},y^{(2)},...,y^{(m)}|x^{(1)},x^{(2)},...,x^{(m)})),$$
 is a subset of $m^2$-tuples of
non-negative numbers $R^{(i,j)}$ for $i,j\in [m]$ defined as
follows: consider the set of all sets $\mathcal{W}^{(1)},
\mathcal{W}^{(2)}, ..., \mathcal{W}^{(m)}$, functions
$f^{(i)}(W^{(1)}, W^{(2)}, ..., W^{(m)})$ ($1\leq i \leq m$) having
the special form of
$$(f^{(i,1)}(W^{(1)}), f^{(i,2)}(W^{(2)}), ...,
f^{(i,m)}(W^{(m)})),$$ the distortion functions
$\Delta^{(i)}(m^{(i)},m^{'(i)})$ (for $1\leq i \leq m$) being equal
to the indicator function $\textbf{1}[m^{(i)}\neq m^{'(i)}]$,
$D^{(i)}$ being set to be zero for all $1\leq i \leq m$ and
admissible sources $p(w^{(1)}, w^{(2)},...,w^{(m)})$ for which
$f^{(i,j)}(W^{(j)})$'s are mutually independent of each other. The
capacity region is then taken to be the set of all achievable
$R^{(i,j)}=H(f^{(j,i)}(W^{(i)}))$ (for $i,j\in [m]$) given the above
constraints. Intuitively speaking, $R^{(i,j)}$ is the communication
rate from $i^{th}$ party to the $j^{th}$ party.
\end{definition}

\begin{definition} \label{def:oplus} For any natural number $c$ and any two sets of
points $K$ and $L$ in $\mathbb{R}_{+}^c $, let $K \oplus L$ refer to
their Minkowski sum: $K \oplus L = \{v_1+v_2: v_1\in K, v_2\in L\}$.
For any real number $r$, let $r\times K=\{r\cdot v_1: v_1\in K\}$.
We also define $\frac{K}{r}$ as the set formed by shrinking $K$
through scaling each point of it by a factor $\frac{1}{r}$. Note
that in general $r\times K\neq (r_1\times K)\oplus (r_2\times K)$
when $r=r_1+r_2$ but this is true when $K$ is a convex
set.\end{definition}

\begin{definition} \label{def:downset} For any two points $\overrightarrow{v_1}$ and
$\overrightarrow{v_2}$ in $\mathbb{R}_{+}^c$, we say
$\overrightarrow{v_1}\geq \overrightarrow{v_2}$ if and only if each
coordinate of $\overrightarrow{v_1}$ is greater than or equal to the
corresponding coordinate of $\overrightarrow{v_2}$. For any two sets
of points $A$ and $B$ in $\mathbb{R}_{+}^c$, we say $A\leq B$ if and
only if for any point $\overrightarrow{a}\in A$, there exists a
point $\overrightarrow{b}\in B$ such that $\overrightarrow{a}\leq
\overrightarrow{b}$. For a set $A \in \mathbb{R}_{+}^c$, the
down-set $\Pi(A)$ is defined as: $\Pi(A)=\{\overrightarrow{v}
\in \mathbb{R}_{+}^c:\ \overrightarrow{v} \leq \overrightarrow{w}
\mbox{ for some } \overrightarrow{w} \in A\}$.\end{definition}

\begin{definition} \label{def:DefinitionOfPsi}Given a specific network architecture
$q(y^{(1)}, y^{(2)}, ..., y^{(m)}|x^{(1)}, x^{(2)}, ..., x^{(m)})$,
and
 the source marginal distribution $p(w^{(1)},
w^{(2)},...,w^{(m)})$, it may be
 possible to find properties that the inputs to the multiterminal
 network throughout the communication
  satisfy. For
 instance in an interference channel or a multiple access channel with no output feedback, if the transmitters observe
 independent messages, the random variables representing their
 information stay independent of each other throughout the communication. This is because the transmitters
 neither interact nor receive any feedback
 from the outputs. Other constraints on the inputs to the
 network might come from practical requirements such as a maximum instantaneous power used up by one or a group of nodes in each stage of the communication.
  Given a multiterminal network $q(y^{(1)}, y^{(2)}, ..., y^{(m)}|x^{(1)}, x^{(2)}, ..., x^{(m)})$
  and assuming that $\mathcal{X}^{(i)}$ ($i\in [m]$) is the set $X^{(i)}$ is taking value from,
   let $\Psi$ be a set of
joint distributions on $\mathcal{X}^{(1)} \times \mathcal{X}^{(2)}
\times \mathcal{X}^{(3)} \times ...\times \mathcal{X}^{(m)}$ for
which the following guarantee exists: for any communication
protocol, the inputs to the multiterminal network at each time stage
have a joint distribution belonging to the set $\Psi$. Such a set
will be called a \emph{permissible set} of input distributions. Some
of the results below will be stated in terms of this nebulously
defined region $\Psi$. To get explicit results, simply replace
$\Psi$ by the set of all probability distributions
 on $\mathcal{X}^{(1)} \times \mathcal{X}^{(2)} \times
\mathcal{X}^{(3)} \times ...\times \mathcal{X}^{(m)}$.
\end{definition}

\section{Statement of the results}\label{Section:StatementOfResults}
\label{Section:CutSet}
\begin{theorem}\label{Thm:CutSet} Given any GMN
$q(y^{(1)},y^{(2)},...,y^{(m)}|x^{(1)},x^{(2)},...,x^{(m)})$, a
sequence of non-negative real numbers $D^{(i)}$ ($i \in [m]$), an
 arbitrary admissible source $W^{(i)}$ ($i\in [m]$), and a
 permissible
 set of input distributions of the network $\Psi$,
 there exists
\begin{itemize}
  \item joint distribution
$q(x^{(1)},x^{(2)},...,x^{(m)}, z)$ where size of the alphabet set
of $Z$ is $2^m-1$ and furthermore $q(x^{(1)},x^{(2)},...,x^{(m)}|z)$
belongs to $\Psi$ for any value $z$ that the random variable $Z$
might take;
  \item joint distribution $p(\widehat{m}^{(1)}, \widehat{m}^{(2)}, ...,
\widehat{m}^{(m)}, w^{(1)}, w^{(2)}, ..., w^{(m)})$ where the
average distortion between $M^{(i)}=f^{(i)}(W^{(1)}, W^{(2)}, ...,
W^{(m)})$ and $\widehat{M}^{(i)}$ is less than or equal to
$D^{(i)}$, i.e. $\Delta^{(i)}(M^{(i)}, \widehat{M}^{(i)})\leq
D^{(i)}$,
\end{itemize}
  such that for any arbitrary $T\subset
[m]$ the following inequality holds:
\[
I\big(W^{(i)}:i\in T \ \ \textbf{;} \ \ \widehat{M}^{(j)}:j\in
T^c|W^{(j)}:j\in T^c\big)\leq I\big(X^{(i)}:i\in T \ \ \textbf{;} \
\ Y^{(j)}:j\in T^c|X^{(j)}:j\in T^c, Z\big),\] where
$Y^{(1)},Y^{(2)},...,Y^{(m)}, X^{(1)},X^{(2)},...,X^{(m)}$ and $Z$
are jointly distributed according to
$$q(y^{(1)},y^{(2)},...,y^{(m)}|x^{(1)},x^{(2)},...,x^{(m)})\cdot q(x^{(1)},x^{(2)},...,x^{(m)},z).$$
Note that here the following Markov chain holds:
$$Z-X^{(1)},X^{(2)},...,X^{(m)}-Y^{(1)},Y^{(2)},...,Y^{(m)}.$$
\end{theorem}
\begin{discussion}
The fact that the expressions on both sides of the above inequality
are of the same form is suggestive. To any given channel
$q(y^{(1)},y^{(2)},...,y^{(m)}|x^{(1)},x^{(2)},...,x^{(m)})$ and
input distribution $q(x^{(1)},x^{(2)},...,x^{(m)})$, assign the
down-set of a vector in $\mathbb{R}_{+}^{2^m}$ whose $k^{th}$
coordinate is defined as \[I\big(X^{(i)}:i\in T_k \ \ \textbf{;} \ \
Y^{(j)}:j\in T_k^c|X^{(j)}:j\in T_k^c\big),\] where $T_k$ is defined
as follows: there are $2^m$ subsets of $[m]$; take an arbitrary
ordering of these sets and take $T_k$ to be the $k^{th}$ subset in
that ordering (though not required but for the sake of consistency
with the notation used in the proof of the theorem assume that
$T_{2^k-1}$ and $T_{2^k}$ are the empty set and the full set
respectively). Next, to any channel
$q(y^{(1)},y^{(2)},...,y^{(m)}|x^{(1)},x^{(2)},...,x^{(m)})$ and a
set of permissible input distributions, we assign a region by taking
the convex hull of the union over all permissible input
distributions, of the region associated to the channel and the
varying input distribution. A channel is said to be weaker than
another channel if the region associated to the first channel is
contained in the region associated to the second channel.

Intuitively speaking, given a communication task one can consider a
virtual channel whose inputs and outputs represent, roughly
speaking, the raw and acceptable information objectives at the $m$
parties. Furthermore, let the only permissible input distribution
for this virtual channel to be one given by the statistical
description of the raw information of the parties. More
specifically, given any $p(\widehat{m}^{(1)}, ...,
\widehat{m}^{(m)}, w^{(1)}, ..., w^{(m)})$ such that $
\Delta^{(i)}(M^{(i)}, \widehat{M}^{(i)})\leq D^{(i)}$ holds,
consider the virtual channel $p(\widehat{m}^{(1)},
\widehat{m}^{(2)}, ..., \widehat{m}^{(m)}| w^{(1)}, w^{(2)}, ...,
w^{(m)})$ and the input distribution $p(w^{(1)}, w^{(2)}, ...,
w^{(m)})$. The inputs of this virtual channel, i.e. $W^{(1)},
W^{(2)}, ..., W^{(m)}$, and its outputs, i.e. $\widehat{M}^{(1)},
\widehat{M}^{(2)}, ..., \widehat{M}^{(m)}$, can be understood as the
raw information and acceptable information objectives at the $m$
parties. The region associated to the virtual channel
$p(\widehat{m}^{(1)}, ..., \widehat{m}^{(m)}| w^{(1)}, ...,
w^{(m)})$ and the input distribution $p(w^{(1)}, w^{(2)}, ...,
w^{(m)})$ would be the down-set of a vector in
$\mathbb{R}_{+}^{2^m}$ whose $k^{th}$ coordinate is defined as
\[I\big(W^{(i)}:i\in T \ \ \textbf{;} \ \ \widehat{M}^{(j)}:j\in
T^c|W^{(j)}:j\in T^c\big).\] Theorem \ref{Thm:CutSet} is basically
saying that this region associated to this virtual channel and the
corresponding input distribution should be included inside the
region associated to the channel
$q(y^{(1)},y^{(2)},...,y^{(m)}|x^{(1)},x^{(2)},...,x^{(m)})$. Here
the complexity of transmission of functions of correlated messages
is effectively translated into the performance region of a virtual
channel at a given input distribution. This virtual channel at the
given input distribution must be, in the above mentioned sense,
weaker than any physical channel fit for the communication problem.
\end{discussion}
\begin{corollary} Given any GMN
$q(y^{(1)},y^{(2)},...,y^{(m)}|x^{(1)},x^{(2)},...,x^{(m)})$, the
following region forms an outer bound on the independent messages
zero distortion capacity region (see Definition
\ref{def:capacityregion}) of the network:
\[\bigcup_{\scriptsize{\begin{array}{l}
              q(x^{(1)},x^{(2)},...,x^{(m)},z) \mbox{ such that for any }z\\
              q(x^{(1)},x^{(2)},...,x^{(m)}|z)\in \Psi  \mbox{ and}\\
              \mbox{size of the alphabet set of $Z$ is $2^m-1$}
            \end{array}}}
\bigg\{ \mbox{non-negative }R^{(i,j)} \mbox{ for } i,j\in [m]
\mbox{: for any arbitrary }T\subset [m]\]\[ \sum_{i\in T, j\in
T^c}R^{(i,j)} \leq I\big(X^{(i)}:i\in T \ \ \textbf{;} \ \
Y^{(j)}:j\in T^c|X^{(j)}:j\in T^c, Z\big)\]\[\mbox{is
satisfied.\bigg\}},\] where $Y^{(1)},Y^{(2)},...,Y^{(m)},
X^{(1)},X^{(2)},...,X^{(m)}$ and $Z$ are jointly distributed
according to
$$q(y^{(1)},y^{(2)},...,y^{(m)}|x^{(1)},x^{(2)},...,x^{(m)})\cdot q(x^{(1)},x^{(2)},...,x^{(m)},z).$$
\end{corollary}
\begin{remark} This bound is sometimes tight; for instance it is tight
for a multiple access channel with independent source messages when
$\Psi$ is taken to be the set of all mutually independent input
distributions.
\end{remark}
\begin{remark}This bound reduces to the traditional cut-set bound when
$\Psi$ is taken to be the set of all input distributions, and
$I\big(X^{(i)}:i\in T \ \ \textbf{;} \ \ Y^{(i)}:i\in
T^c|X^{(i)}:i\in T^c, Z\big)$ is bounded from above by
\footnote{This is valid because $I\big(X^{(i)}:i\in T \ \ \textbf{;}
\ \ Y^{(j)}:j\in T^c|X^{(j)}:j\in T^c, Z\big)= H\big(Y^{(j)}:j\in
T^c|X^{(j)}:j\in T^c, Z\big)-H\big(Y^{(j)}:j\in T^c|X^{(i)}:i\in
[m], Z\big)=H\big(Y^{(j)}:j\in T^c|X^{(j)}:j\in T^c,
Z\big)-H\big(Y^{(j)}:j\in T^c|X^{(i)}:i\in [m]\big)\leq
H\big(Y^{(j)}:j\in T^c|X^{(j)}:j\in T^c\big)-H\big(Y^{(j)}:j\in
T^c|X^{(i)}:i\in [m]\big)=I\big(X^{(i)}:i\in T \ \ \textbf{;} \ \
Y^{(j)}:j\in T^c|X^{(j)}:j\in T^c\big).$}
$$I\big(X^{(i)}:i\in T \ \ \textbf{;} \ \ Y^{(j)}:j\in
T^c|X^{(j)}:j\in T^c\big).$$\end{remark}

\subsection{The Main Lemma}
During the simulation of the code, the information of the parties
begins from the $i^{th}$ party having $W_{1:n}^{(i)}$ and gradually
evolves over time with the usage of the network. At the $j^{th}$
stage, the $i^{th}$ party has $W_{1:n}^{(i)}Y_{1:j}^{(i)}$. We
represent the information state of the whole system at the $j^{th}$
stage by the virtual channel $p(w_{1:n}^{(1)}y_{1:j}^{(1)},...,
w_{1:n}^{(m)}y_{1:j}^{(m)}|w_{1:n}^{(1)},..., w_{1:n}^{(m)})$ and
the input distribution $p(w_{1:n}^{(1)},..., w_{1:n}^{(m)})$. In
order to quantify the information state, we map the information
state to a subset of $\mathbb{R}_{+}^{c}$ ($c$ is a natural number)
using a function $\phi(.)$. A formal definition of $\phi$ and the
properties we require it to satisfy are as follows:

Let $\phi (p(y^{(1)},...,y^{(m)}| x^{(1)}, ...,x^{(m)}), \Psi)$ be a
function that takes as input an arbitrary $m$-input/$m$-output GMN
and a subset of probability distributions on the inputs of this
network and returns a subset of $\mathbb{R}_{+}^{c}$ where $c$ is a
natural number. $\phi(.)$ is thus a function from the set of all
conditional probability distributions defined on finite sets and a
corresponding set of input distributions, to subsets of
$\mathbb{R}_{+}^{c}$.

Assume that the function $\phi(.)$ satisfies the following three
properties. The intuitive description of the properties is provided
after their formal statement. Please see Definitions \ref{def:oplus}
and \ref{def:downset} for the notations used.
\begin{enumerate}
        \item Assume that the conditional distribution
        $p(y^{(1)}y^{'(1)},y^{(2)}y^{'(2)},...,y^{(m)}y^{'(m)}|x^{(1)},x^{(2)},...,x^{(m)})$
        satisfies the following
\begin{eqnarray*}&p(y^{(1)}y^{'(1)},y^{(2)}y^{'(2)},...,y^{(m)}y^{'(m)}|x^{(1)},...,x^{(m)})
\\&=
p(y^{(1)},y^{(2)},...,y^{(m)}|x^{(1)},x^{(2)},...,x^{(m)})\cdot\\&
p(y^{'(1)},y^{'(2)},...,y^{'(m)}|x^{'(1)},x^{'(2)},...,x^{'(m)}),\end{eqnarray*}
where $X^{'(i)}$ is a deterministic function of $Y^{(i)}$ (i.e.
$H(X^{'(i)}|Y^{(i)})=0$ ($i\in [m]$)). Random variable $X^{'(i)}$
(for $i\in [m]$) is assumed to take value from set
$\mathcal{X}^{'(i)}$. Take an arbitrary input distribution $q(x_1,
x_2,..., x_m)$. This input distribution, together with the
conditional distribution
$p(y^{(1)},y^{(2)},...,y^{(m)}|x^{(1)},x^{(2)},...,x^{(m)})$, impose
a joint distribution $q(x^{'(1)},x^{'(2)},...,x^{'(m)})$ on
$\\(X^{'(1)},X^{'(2)},...,X^{'(m)})$. Then the following constraint
needs to be satisfied for any arbitrary set $\Psi$ of joint
distributions on $\mathcal{X}^{'(1)} \times \mathcal{X}^{'(2)}
\times \cdot\cdot\cdot\times \mathcal{X}^{'(m)}$ that contains
$q(x^{'(1)},x^{'(2)},...,x^{'(m)})$: \begin{eqnarray*}&\phi
\bigg(p(y^{(1)}y^{'(1)},...,y^{(m)}y^{'(m)}|x^{(1)},...,x^{(m)}) \\&
, \{q(x_1,...,x_m)\}\bigg)\subseteq\\& \phi\big(
p(y^{(1)},...,y^{(m)}|x^{(1)},...,x^{(m)}), \{q(x_1,...,x_m)\}\big)
\\&\oplus\
\phi\big(p(y^{'(1)},y^{'(2)},...,y^{'(m)}|x^{'(1)},...,x^{'(m)}),
\Psi\big).\end{eqnarray*}
    \item Assume that
        \begin{eqnarray*}p(y^{(1)},...,y^{(m)}|x^{(1)},...,x^{(m)})=\prod_{i=1}^m \mathbf{1}[y^{(i)}=x^{(i)}].\end{eqnarray*}
        Then we require that for any input
distribution $q(x_1, x_2,..., x_m)$, the set
$$\phi\big(p(y^{(1)},...,y^{(m)}|x^{(1)},...,x^{(m)}),
\{q(x_1,...,x_m)\}\big)$$ contains only the origin in
$\mathbb{R}^{c}$.
    \item Assume that
    \begin{eqnarray*}&p(z^{(1)},...,z^{(m)}, y^{(1)},...,y^{(m)}|x^{(1)},...,x^{(m)})
    =\\& p(y^{(1)},...,y^{(m)}|x^{(1)},...,x^{(m)})\prod_{i=1}^m p(z_i|y_i).\end{eqnarray*}
     Then we require that for any input
distribution $q(x_1, x_2,..., x_m)$,\[\phi \big(p(z^{(1)}, ...,
z^{(m)}|x^{(1)},...,x^{(m)}),
\{q(x_1,...,x_m)\}\big)\subseteq\]\[\phi \big(p(y^{(1)}, ...,
y^{(m)}|x^{(1)},...,x^{(m)}), \{q(x_1,...,x_m)\}\big).\]
\end{enumerate}

The first condition is intuitively saying that additional use of the
channel
$$p(y^{'(1)},y^{'(2)},...,y^{'(m)}|x^{'(1)},x^{'(2)},...,x^{'(m)})$$
can expand $\phi(.)$ by at most
$$\phi\big(p(y^{'(1)},y^{'(2)},...,y^{'(m)}|x^{'(1)},x^{'(2)},...,x^{'(m)}),
\Psi\big).$$ The second condition is intuitively saying that
$\phi(.)$ vanishes if the parties are unable to communicate, that is
each party receives exactly what it puts at the input of the
channel. The third condition is basically saying that making a
channel weaker at each party can not cause $\phi(.)$ expand.

\begin{lemma} \label{Thm:MainThm} For any function $\phi(.)$
satisfying the above three properties, and for any multiterminal
network
$$q(y^{(1)},y^{(2)},...,y^{(m)}|x^{(1)},x^{(2)},...,x^{(m)}),$$
distortions $D^{(i)}$ and arbitrary admissible source $W^{(i)}$ ($i
\in [m]$), positive $\epsilon$ and $(n)$-code satisfying the
distortion constraints and a permissible set $\Psi$ of input
distributions, we have (for the definition of multiplication of a
set by a real number see Definition \ref{def:oplus}):
\[
\phi\big( p(\widehat{m}_{1:n}^{(1)}, ...,
\widehat{m}_{1:n}^{(m)}|w_{1:n}^{(1)},...,w_{1:n}^{(m)}),
\{p(w_{1:n}^{(1)},...,w_{1:n}^{(m)})\}\big)\subseteq\]
\[ n\times Convex\ Hull\big\{\phi\big(q(y^{(1)},...,y^{(m)}|x^{(1)},...,x^{(m)}), \Psi\big)\big\},
\]
where $W_{1:n}^{(i)}\ (i \in [m])$ are the messages observed at the
nodes; $\widehat{M}_{1:n}^{(i)}\ (i \in [m])$ are the
reconstructions by the parties at the end of the communication
satisfying
$$\mathbb{E}\bigg[ \Delta_n^{(i)}\bigg((\widehat{m}_{1:n}^{(i)},
m_{1:n}^{(i)}\bigg)\bigg]\leq D^{(i)}+\epsilon,$$ for any $i\in
[m]$.
\end{lemma}

\section{Proofs} \label{Section:Proofs}

\begin{proof}[Proof of Lemma \ref{Thm:MainThm}] Let random variables $X_k^{(i)}$ and $Y_k^{(i)}$ ($k\in
[n]$, $i\in [m]$) respectively represent the inputs to the
multiterminal network and the outputs at the nodes of the network.
We have:
\begin{equation}\label{Eq:Thm102}
\phi\big( p(\widehat{m}_{1:n}^{(1)}, \widehat{m}_{1:n}^{(2)}, ...,
\widehat{m}_{1:n}^{(m)}|w_{1:n}^{(1)},w_{1:n}^{(2)},...,w_{1:n}^{(m)}),
\{p(w_{1:n}^{(1)},w_{1:n}^{(2)},...,w_{1:n}^{(m)})\}\big)\subseteq
\end{equation}
\begin{equation}\label{Eq:Thm1i2}
\phi\big(
p(w_{1:n}^{(1)}y_{1:n}^{(1)},w_{1:n}^{(2)}y_{1:n}^{(2)},...,w_{1:n}^{(m)}y_{1:n}^{(m)}|w_{1:n}^{(1)},w_{1:n}^{(2)},...,w_{1:n}^{(m)}),
\{p(w_{1:n}^{(1)},w_{1:n}^{(2)},...,w_{1:n}^{(m)})\}\big)\subseteq\end{equation}
\[\phi\big(
p(w_{1:n}^{(1)}y_{1:n-1}^{(1)},w_{1:n}^{(2)}y_{1:n-1}^{(2)},...,w_{1:n}^{(m)}y_{1:n-1}^{(m)}|w_{1:n}^{(1)},w_{1:n}^{(2)},...,w_{1:n}^{(m)}),
\{p(w_{1:n}^{(1)},w_{1:n}^{(2)},...,w_{1:n}^{(m)})\}\big)\oplus\]\[
\phi(q(y_n^{(1)},y_n^{(2)},...,y_n^{(m)}|x_n^{(1)},x_n^{(2)},...,x_n^{(m)}),
\Psi) \subseteq\]
\[\phi\big(
p(w_{1:n}^{(1)}y_{1:n-2}^{(1)},w_{1:n}^{(2)}y_{1:n-2}^{(2)},...,w_{1:n}^{(m)}y_{1:n-2}^{(m)}|w_{1:n}^{(1)},w_{1:n}^{(2)},...,w_{1:n}^{(m)}),
\{p(w_{1:n}^{(1)},w_{1:n}^{(2)},...,w_{1:n}^{(m)})\}\big)\oplus\]
\[
\phi(q(y_{n-1}^{(1)},y_{n-1}^{(2)},...,y_{n-1}^{(m)}|x_{n-1}^{(1)},x_{n-1}^{(2)},...,x_{n-1}^{(m)}),
\Psi) \oplus\]
\[
\phi(q(y_n^{(1)},y_n^{(2)},...,y_n^{(m)}|x_n^{(1)},x_n^{(2)},...,x_n^{(m)}),
\Psi) \subseteq\]
\[
\cdot \cdot \cdot\subseteq
\]
\[\phi\big(
p(w_{1:n}^{(1)},w_{1:n}^{(2)},...,w_{1:n}^{(m)}|w_{1:n}^{(1)},w_{1:n}^{(2)},...,w_{1:n}^{(m)}),
\{p(w_{1:n}^{(1)},w_{1:n}^{(2)},...,w_{1:n}^{(m)})\}\big)\oplus\]
\[\phi(q(y_{1}^{(1)},y_{1}^{(2)},...,y_{1}^{(m)}|x_{1}^{(1)},x_{1}^{(2)},...,x_{1}^{(m)}),
\Psi) \oplus\]
\[
\phi(q(y_{2}^{(1)},y_{2}^{(2)},...,y_{2}^{(m)}|x_{2}^{(1)},x_{2}^{(2)},...,x_{2}^{(m)}),
\Psi) \oplus \cdot \cdot \cdot\]
\[
\phi(q(y_{n-1}^{(1)},y_{n-1}^{(2)},...,y_{n-1}^{(m)}|x_{n-1}^{(1)},x_{n-1}^{(2)},...,x_{n-1}^{(m)}),
\Psi) \oplus\]
\begin{equation}\label{Eq:Thm1ii2}
\phi(q(y_n^{(1)},y_n^{(2)},...,y_n^{(m)}|x_n^{(1)},x_n^{(2)},...,x_n^{(m)}),
\Psi) \subseteq\end{equation}
\[\phi(
q(y_{1}^{(1)},y_{1}^{(2)},...,y_{1}^{(m)}|x_{1}^{(1)},x_{1}^{(2)},...,x_{1}^{(m)}),
\Psi)\oplus\]
\[
\phi(q(y_{2}^{(1)},y_{2}^{(2)},...,y_{2}^{(m)}|x_{2}^{(1)},x_{2}^{(2)},...,x_{2}^{(m)}),
\Psi) \oplus \cdot \cdot \cdot\]
\[
\phi(q(y_{n-1}^{(1)},y_{n-1}^{(2)},...,y_{n-1}^{(m)}|x_{n-1}^{(1)},x_{n-1}^{(2)},...,x_{n-1}^{(m)}),
\Psi) \oplus\]
\begin{equation}\label{Eq:Thm1iii2}
\phi(q(y_n^{(1)},y_n^{(2)},...,y_n^{(m)}|x_n^{(1)},x_n^{(2)},...,x_n^{(m)}),
\Psi) \subseteq\end{equation}
\[ n\times Convex\ Hull
\big\{\phi(q(y^{(1)},y^{(2)},...,y^{(m)}|x^{(1)},x^{(2)},...,x^{(m)}),
\Psi)\big\},
\]
where in equation \ref{Eq:Thm102} we have used property (3); in
equation \ref{Eq:Thm1i2} we have used property (1) because
$$p(w_{1:n}^{(1)}y_{1:n}^{(1)},w_{1:n}^{(2)}y_{1:n}^{(2)},...,w_{1:n}^{(m)}y_{1:n}^{(m)}|w_{1:n}^{(1)},w_{1:n}^{(2)},...,w_{1:n}^{(m)})=$$
$$p(w_{1:n}^{(1)}y_{1:n-1}^{(1)},w_{1:n}^{(2)}y_{1:n-1}^{(2)},...,w_{1:n}^{(m)}y_{1:n-1}^{(m)}|w_{1:n}^{(1)},w_{1:n}^{(2)},...,w_{1:n}^{(m)})\cdot
p(y_{n}^{(1)},...,y_{n}^{(m)}|x_{n}^{(1)},x_{n}^{(2)},...,x_{n}^{(m)})$$
 and furthermore $H(X_{n}^{(i)}|W_{1:n}^{(i)}Y_{1:n-1}^{(i)})=0$ for all $i\in [m]$, and
that
$$p(y_{n}^{(1)},y_{n}^{(2)},...,y_{n}^{(m)}|x_{n}^{(1)},x_{n}^{(2)},...,x_{n}^{(m)})=q(y_{n}^{(1)},y_{n}^{(2)},...,y_{n}^{(m)}|x_{n}^{(1)},x_{n}^{(2)},...,x_{n}^{(m)}).$$
The definition of permissible sets implies that the joint
distribution $p(x_{n}^{(1)},x_{n}^{(2)},...,x_{n}^{(m)})$ is in
$\Psi$; in equation \ref{Eq:Thm1ii2} we have used property (2). In
equation \ref{Eq:Thm1iii2}, we first note that the conditional
distributions
$$q(y_i^{(1)},y_i^{(2)},...,y_i^{(m)}|x_i^{(1)},x_i^{(2)},...,x_i^{(m)})$$
for $i=1,2,...,n$ are all the same. We then observe that whenever
$\overrightarrow{v_i} \in \phi(q(y^{(1)},
...,y^{(m)}|x^{(1)},...,x^{(m)}), \Psi)$ for $i \in [n]$, their
average, $\frac{1}{n}\sum_{i=1}^n \overrightarrow{v_i}$ falls in the
convex hull of $ \phi(q(y^{(1)},...,y^{(m)}|x^{(1)},...,x^{(m)}),
\Psi)$.
\end{proof}

\begin{proof}[Proof of Theorem \ref{Thm:CutSet}] The inequalities always hold
for the extreme cases of the set $T$ being either empty or $[m]$.
So, it is sufficient to consider only those subsets of $[m]$ that
are neither empty nor equal to $[m]$. Take an arbitrary $\epsilon >
0$ and an $(n)$-code satisfying the average distortion condition
$D^{(i)}$ (for all $i \in [m]$) over the channel
$q(y^{(1)},y^{(2)},...,y^{(m)}|x^{(1)},x^{(2)},...,x^{(m)})$. Let
random variables $X_k^{(i)}$ and $Y_k^{(i)}$ ($k\in [n]$, $i\in
[m]$) respectively represent the inputs to the multiterminal network
and the outputs at the nodes of the network. Also assume that
$W_{1:n}^{(i)}\ (i \in [m])$ are the messages observed at the nodes.
Let $\widehat{M}_{1:n}^{(i)}\ (i \in [m])$ be the reconstructions
 by the parties at the end of the communication satisfying
$$\mathbb{E}\bigg[ \Delta_n^{(i)}\bigg((\widehat{m}_{1:n}^{(i)},
m_{1:n}^{(i)}\bigg)\bigg]\leq D^{(i)}+\epsilon,$$ for any $i\in
[m]$. Lastly, let $\Psi$ be a permissible set of input
distributions.

We define a function $\phi(.)$ as follows: for any conditional
distribution
$p(y^{(1)},y^{(2)},...,y^{(m)}|x^{(1)},x^{(2)},...,x^{(m)})$ and an
arbitrary set $\Psi$ of distributions on $\mathcal{X}^{(1)}\times
\mathcal{X}^{(2)}\times \cdot\cdot\cdot \mathcal{X}^{(m)}$, let
\begin{equation}\label{Eq:DefinitionOfHardPhi}
\phi(p(y^{(1)},y^{(2)},...,y^{(m)}|x^{(1)},x^{(2)},...,x^{(m)}),
\Psi))=\end{equation}$$\bigcup_{p(x^{(1)},x^{(2)},...,x^{(m)})\in
\Psi}
\varphi\big(p(y^{(1)},y^{(2)},...,y^{(m)}|x^{(1)},x^{(2)},...,x^{(m)})p(x^{(1)},x^{(2)},...,x^{(m)})\big),$$
 where $\varphi(p(y^{(1)},y^{(2)},...,y^{(m)},x^{(1)},x^{(2)},...,x^{(m)}))$ is defined as the down-set
\footnote{For the definition of
 a down-set see Definition
\ref{def:downset}} of a vector of size $c=2^m-2$ whose $k^{th}$
coordinate equals $I\big(X^{(i)}:i\in T_k \ \ \textbf{;} \ \
Y^{(j)}:j\in (T_k)^c|X^{(j)}:j\in (T_k)^c\big)$ where $T_k$ is
defined as follows: there are $2^m-2$ subsets of $[m]$ that are
neither empty nor equal to $[m]$. Take an arbitrary ordering of
these sets and take $T_k$ to be the $k^{th}$ subset in that
ordering.

In appendices \ref{Appendix:ExtensionOne},
\ref{Appendix:ExtensionTwo} and \ref{Appendix:ExtensionThree}, we
 verify that $\phi(.)$ satisfies the
three properties of Lemma \ref{Thm:MainThm} for the choice of
$c=2^m-2$. Lemma \ref{Thm:MainThm} thus implies that (for the
definition of multiplication of a set by a real number see
Definition \ref{def:oplus}):
\[
\phi\big( p(\widehat{m}_{1:n}^{(1)},
\widehat{m}_{1:n}^{(2)},...,\widehat{m}_{1:n}^{(m)}|w_{1:n}^{(1)},w_{1:n}^{(2)},...,w_{1:n}^{(m)}),\{p(w_{1:n}^{(1)},w_{1:n}^{(2)},...,w_{1:n}^{(m)})\}\big)=\]
\[
\varphi\big( p(\widehat{m}_{1:n}^{(1)},
\widehat{m}_{1:n}^{(2)},...,\widehat{m}_{1:n}^{(m)}|w_{1:n}^{(1)},w_{1:n}^{(2)},...,w_{1:n}^{(m)})p(w_{1:n}^{(1)},w_{1:n}^{(2)},...,w_{1:n}^{(m)})\big)\subseteq\]
\[ n\times Convex\ Hull\big\{\phi(q(y^{(1)},y^{(2)},...,y^{(m)}|x^{(1)},x^{(2)},...,x^{(m)}), \Psi)\big\}.
\]

According to the Carath\'{e}odory theorem, every point inside the
convex hull of
$$\phi(q(y^{(1)},y^{(2)},...,y^{(m)}|x^{(1)},x^{(2)},...,x^{(m)}),
\Psi)$$ can be written as a convex combination of $c+1=2^m-1$ points
in the set. Corresponding to the $i^{th}$ point in the convex
combination ($i\in [2^m-1]$) is an input distribution
$q_i(x^{(1)},x^{(2)},...,x^{(m)})$ such that the point lies in
$$\varphi\big(q(y^{(1)},y^{(2)},...,y^{(m)}|x^{(1)},x^{(2)},...,x^{(m)})q_i(x^{(1)},x^{(2)},...,x^{(m)})\big).$$
Let $p(x^{(1)},x^{(2)},...,x^{(m)},z)=p(z)\cdot
q_z(x^{(1)},x^{(2)},...,x^{(m)})$ where $Z$ is a random variable
defined on the set $\{1,2,3,...,2^m-1\}$, taking value $i$ with
probability equal to the weight associated to the $i^{th}$ point in
the above convex combination. The convex hull of
$\phi(q(y^{(1)},y^{(2)},...,y^{(m)}|x^{(1)},x^{(2)},...,x^{(m)}),
\Psi)$ is therefore included in (see Definition \ref{def:oplus} for
the
definition of the summation used here): \[
\bigcup_{\scriptsize{\begin{array}{l}
              q(x^{(1)},x^{(2)},...,x^{(m)},z) \mbox{ such that for any }z\\
              q(x^{(1)},x^{(2)},...,x^{(m)}|z)\in \Psi  \mbox{ and}\\
              \mbox{size of the alphabet set of $Z$ is $2^m-1$}
            \end{array}}}
\sum_{z}p(z)\times
\varphi\big(q(y^{(1)},...,y^{(m)}|x^{(1)},...,x^{(m)})q(x^{(1)},...,x^{(m)}|z)\big).\]
Conversely, the above set only involves convex combination of points
in $\phi(q(y^{(1)},...,y^{(m)}|x^{(1)},...,x^{(m)}), \Psi)$ and
hence is always contained in the convex hull of
$\phi(q(y^{(1)},...,y^{(m)}|x^{(1)},...,x^{(m)}), \Psi)$. Therefore
it must be equal to the convex hull region.

Hence,
\[
\varphi\big(
p(\widehat{m}_{1:n}^{(1)},\widehat{m}_{1:n}^{(2)},...,\widehat{m}_{1:n}^{(m)}|w_{1:n}^{(1)},w_{1:n}^{(2)},...,w_{1:n}^{(m)})p(w_{1:n}^{(1)},w_{1:n}^{(2)},...,w_{1:n}^{(m)})\big)\subseteq\]
\[n \times \bigcup_{\scriptsize{\begin{array}{l}
              q(x^{(1)},x^{(2)},...,x^{(m)},z) \mbox{ such that for any }z\\
              q(x^{(1)},x^{(2)},...,x^{(m)}|z)\in \Psi  \mbox{ and}\\
              \mbox{size of the alphabet set of $Z$ is $2^m-1$}
            \end{array}}}
\sum_{z}p(z)\times
\varphi\big(q(y^{(1)},...,y^{(m)}|x^{(1)},...,x^{(m)})q(x^{(1)},...,x^{(m)}|z)\big).\]

The set \[ \varphi\big(
p(\widehat{m}_{1:n}^{(1)},\widehat{m}_{1:n}^{(2)},...,\widehat{m}_{1:n}^{(m)}|w_{1:n}^{(1)},w_{1:n}^{(2)},...,w_{1:n}^{(m)})p(w_{1:n}^{(1)},w_{1:n}^{(2)},...,w_{1:n}^{(m)})\big)\]
is by definition the down-set of a vector of length $2^m-2$, denoted
here by $\overrightarrow{v}$, whose $k^{th}$ coordinate is equal to
$$I\big(W_{1:n}^{(i)}:i\in T_k \ \ \textbf{;} \ \
\widehat{M}_{1:n}^{(j)}:j\in (T_k)^c|W_{1:n}^{(j)}:j\in
(T_k)^c\big).$$ The vector $\overrightarrow{v}$ is greater than or
equal to $\overrightarrow{\widetilde{v}}$ whose $k^{th}$ element
equals:\footnote{This is because for any arbitrary random variables
$X^n, Y^n, Z^n$ such that $(X^n, Y^n)$ is $n$ i.i.d. repetition of
$(X, Y)$, we have: $I(X^n;Z^n|Y^n)=nH(X|Y)-H(X^n|Z^nY^n)\geq
\sum_{g=1}^nH(X_g|Y_g)-H(X_g|Y_gZ_g)=\sum_{g=1}^n
I(X_g;Z_g|Y_g)=n\cdot I(X_G;Z_G|GY_G)\geq n\cdot I(X_G;Z_G|Y_G)$
where $G$ is uniform over $\{1,2,...,n\}$ and independent of $(X^n,
Y^n, Z^n)$. Random variables $(X_G, Y_G)$ have the same joint
distribution as $(X, Y)$.}
$$n\cdot I\big(\widetilde{W}^{(i)}:i\in T_k \ \ \textbf{;} \ \
\widetilde{\widehat{M}^{(j)}}:j\in (T_k)^c|\widetilde{W}^{(j)}:j\in
(T_k)^c\big),$$ for some $\widetilde{W}^{(i)}$ and
$\widetilde{\widehat{M}^{(i)}}$ ($i\in [m]$) such that the joint
distribution of $\widetilde{W}^{(i)}$ ($i\in [m]$) is the same as
that of $W^{(i)}$ ($i\in [m]$), and that the average distortion
between $\widetilde{M}^{(i)}=f^{(i)}(\widetilde{W}^{(1)},
\widetilde{W}^{(2)}, ..., \widetilde{W}^{(m)})$ and
$\widetilde{\widehat{M}^{(i)}}$ is less than or equal to
$D^{(i)}+\epsilon$.\footnote{This is because for any arbitrary pair
$(Y^n, Z^n)$, the average distortion between $Y_G$ and $Z_G$ for $G$
uniform over $\{1,2,...,n\}$ and independent of $(Y^n, Z^n)$, is
equal to
$\mathbb{E}[\Delta(Y_G,Z_G)]=\mathbb{E}[\mathbb{E}[\Delta(Y_G,Z_G)|G]]=\sum_{g=1}^{n}
\frac{1}{n}\mathbb{E}[\Delta (Y_g, Z_g)]=\mathbb{E}[\Delta_n (Y^n,
Z^n)]$.} In Appendix \ref{Appendix:ExtensionFive}, we perturb random
variables $\widetilde{\widehat{M}^{(i)}}$ (for $i \in [m]$) and
define random variables $\widetilde{\widehat{M}^{'(i)}}$ (for $i \in
[m]$) such that for every $i\in [m]$, the average distortion between
$\widetilde{\widehat{M}^{'(i)}}$ and $\widetilde{M}^{(i)}$ is less
than or equal to $D^{(i)}$ (rather than $D^{(i)}+\epsilon$ as in the
case of $\widetilde{\widehat{M}^{(i)}}$) and furthermore for every
$k$
$$I\big(\widetilde{W}^{(i)}:i\in T_k \ \ \textbf{;} \ \
\widetilde{\widehat{M}^{'(j)}}:j\in (T_k)^c|\widetilde{W}^{(j)}:j\in
(T_k)^c\big)-O(\tau(\epsilon))\leq$$$$I\big(\widetilde{W}^{(i)}:i\in
T_k \ \ \textbf{;} \ \ \widetilde{\widehat{M}^{(j)}}:j\in
(T_k)^c|\widetilde{W}^{(j)}:j\in (T_k)^c\big),$$ where $\tau(.)$ is
a real-valued function that satisfies the property that
$\tau(\epsilon)\rightarrow 0$ as $\epsilon \rightarrow 0$.

Hence the vector $\overrightarrow{\widetilde{v}}$
 is coordinate by coordinate greater than or equal to a vector $\overrightarrow{\widetilde{v}'}$ whose $k^{th}$ element is
defined as
$$\max\bigg(n\cdot I\big(\widetilde{W}^{(i)}:i\in T_k \ \ \textbf{;} \ \
\widetilde{\widehat{M}^{'(j)}}:j\in (T_k)^c|\widetilde{W}^{(j)}:j\in
(T_k)^c\big)-n\cdot O(\tau(\epsilon)), 0\bigg).$$ The vector
$\overrightarrow{\widetilde{v}'}$ must lie in
\[
\varphi\big(
p((\widehat{m}_{1:n}^{(1)},\widehat{m}_{1:n}^{(2)},...,\widehat{m}_{1:n}^{(m)}|w_{1:n}^{(1)},w_{1:n}^{(2)},...,w_{1:n}^{(m)})p(w_{1:n}^{(1)},w_{1:n}^{(2)},...,w_{1:n}^{(m)})\big),\]
since it is coordinate by coordinate less than or equal to
$\overrightarrow{\widetilde{v}}$. It must therefore also lie in \[n
\times \bigcup_{\scriptsize{\begin{array}{l}
              q(x^{(1)},...,x^{(m)},z) \mbox{ such that for any }z\\
              q(x^{(1)},...,x^{(m)}|z)\in \Psi  \mbox{ and}\\
              \mbox{size of the alphabet set of $Z$ is $2^m-1$}
            \end{array}}}
\sum_{z}p(z)\times
\varphi\big(q(y^{(1)},...,y^{(m)}|x^{(1)},...,x^{(m)})q(x^{(1)},...,x^{(m)}|z)\big).\]
Please note that since $\varphi(.)$ is the down-set of a
non-negative vector, the above Minkowski sum inside the union would
itself be the down-set of a vector.\footnote{This is because for
every two non-negative vectors $\overrightarrow{v_1}$ and
$\overrightarrow{v_2}$, we have
$\lambda\times\Pi(\overrightarrow{v_1})\oplus(1-\lambda)\times\Pi(\overrightarrow{v_2})=\Pi(\lambda\overrightarrow{v_1}+(1-\lambda)\overrightarrow{v_2})$
for any $\lambda \in [0,1]$.} The left hand side can be therefore
written as union over all $q(x^{(1)},x^{(2)},...,x^{(m)},z)$ such
that $q(x^{(1)},x^{(2)},...,x^{(m)}|z)\in \Psi$ for every $z$, of
the down-set of a vector whose $k^{th}$ coordinate equals
$I\big(X^{(i)}:i\in T_k \ \ \textbf{;} \ \ Y^{(j)}:j\in
(T_k)^c|X^{(j)}:j\in (T_k)^c, Z\big)$. Since the
$\overrightarrow{\widetilde{v}'}$ falls inside this union, there
must exist a particular $q(x^{(1)},x^{(2)},...,x^{(m)},z)$ whose
corresponding vector is coordinate by coordinate greater than or
equal to $\overrightarrow{\widetilde{v}'}$. The proof ends by
recalling the definition of $\overrightarrow{\widetilde{v}'}$ and
letting $\epsilon$ converge zero.
\end{proof}
\appendices
\section{Completing the proof of Theorem \ref{Thm:CutSet}}\label{Appendix:CutSet}
\subsection{Checking the first property of Lemma \ref{Thm:MainThm}} \label{Appendix:ExtensionOne}
Given the definition of $\phi(.)$ in equation
\ref{Eq:DefinitionOfHardPhi}, one needs to verify that:
\[\varphi\big(p(y^{(1)}y^{'(1)},...,y^{(m)}y^{'(m)}|x^{(1)},x^{(2)},...,x^{(m)})p(x^{(1)},x^{(2)},...,x^{(m)})\big)
\subseteq\]\[\varphi\big(p(y^{(1)},y^{(2)},...,y^{(m)}|x^{(1)},x^{(2)},...,x^{(m)})p(x^{(1)},x^{(2)},...,x^{(m)})\big)
\oplus\]\[\bigcup_{p(x^{'(1)},x^{'(2)},...,x^{'(m)})\in \Psi}
\varphi\big(p(y^{'(1)},y^{'(2)},...,y^{'(m)}|x^{'(1)},x^{'(2)},...,x^{'(m)})p(x^{'(1)},x^{'(2)},...,x^{'(m)})\big).\]

Take an arbitrary point $\overrightarrow{v}$ inside
$$\varphi\big(p(y^{(1)}y^{'(1)},...,y^{(m)}y^{'(m)}|x^{(1)},x^{(2)},...,x^{(m)})p(x^{(1)},x^{(2)},...,x^{(m)})\big).$$
We would like to prove that there exists $$\overrightarrow{v_1}\in
\varphi
\big(p(y^{(1)},y^{(2)},...,y^{(m)}|x^{(1)},x^{(2)},...,x^{(m)})p(x^{(1)},x^{(2)},...,x^{(m)})\big),$$
and $$\overrightarrow{v_2} \in \varphi
\big(p(y^{'(1)},y^{'(2)},...,y^{'(m)}|x^{'(1)},x^{'(2)},...,x^{'(m)})
p(x^{'(1)},x^{'(2)},...,x^{'(m)})\big),$$ such that
$\overrightarrow{v_1}+\overrightarrow{v_2}\geq \overrightarrow{v}.$

Since $\overrightarrow{v}$ is inside
$$\varphi\big(p(y^{(1)}y^{'(1)},...,y^{(m)}y^{'(m)}|x^{(1)},x^{(2)},...,x^{(m)})p(x^{(1)},x^{(2)},...,x^{(m)})\big),$$
the $k^{th}$ coordinate of $\overrightarrow{v}$ is less than or
equal to $I\big(X^{(i)}:i\in T_k \ \ \textbf{;} \ \
Y^{(j)}Y^{'(j)}:j\in (T_k)^c|X^{(j)}:j\in (T_k)^c\big)$ where $T_k$
is defined as in the proof of Theorem \ref{Thm:CutSet}.

We have:
$$I\big(X^{(i)}:i\in T_k \ \ \textbf{;} \ \ Y^{(j)}Y^{'(j)}:j\in
(T_k)^c|X^{(j)}:j\in (T_k)^c\big)=$$$$I\big(X^{(i)}:i\in T_k \ \
\textbf{;} \ \ Y^{(j)}:j\in (T_k)^c|X^{(j)}:j\in
(T_k)^c\big)+$$$$I\big(X^{(i)}:i\in T_k \ \ \textbf{;} \ \
Y^{'(j)}:j\in (T_k)^c|X^{(j)}:j\in (T_k)^c, Y^{(j)}:j\in (T_k)^c
\big).$$

The second term can be written as:
\begin{equation}\label{Eq:AppendixBEqA}I\big(X^{(i)}:i\in T_k \ \ \textbf{;} \ \
Y^{'(j)}:j\in (T_k)^c|X^{(j)}:j\in (T_k)^c, Y^{(j)}:j\in (T_k)^c
\big)\leq\end{equation}\begin{equation}\label{Eq:AppendixBEqB}
I\big(X^{(i)}X^{'(i)}:i\in T_k \ \ \textbf{;} \ \ Y^{'(j)}:j\in
(T_k)^c|X^{(j)}:j\in (T_k)^c, Y^{(j)}X^{'(j)}:j\in (T_k)^c
\big)=\end{equation}$$I\big(X^{'(i)}:i\in T_k \ \ \textbf{;} \ \
Y^{'(j)}:j\in (T_k)^c|X^{(j)}X^{'(j)}Y^{(j)}:j\in (T_k)^c
\big)+0=$$$$I\big(X^{'(i)}:i\in T_k, X^{(j)}Y^{(j)}:j\in (T_k)^c \ \
\textbf{;} \ \ Y^{'(j)}:j\in (T_k)^c|X^{'(j)}:j\in
(T_k)^c\big)-$$\begin{equation}\label{Eq:AppendixBEqC}
I\big(X^{(j)}Y^{(j)}:j\in (T_k)^c \ \ \textbf{;} \ \ Y^{'(j)}:j\in
(T_k)^c|X^{'(j)}:j\in
(T_k)^c\big)=\end{equation}$$I\big(X^{'(i)}:i\in T_k \ \ \textbf{;}
\ \ Y^{'(j)}:j\in (T_k)^c|X^{'(j)}:j\in (T_k)^c\big)-$$$$
I\big(X^{(j)}Y^{(j)}:j\in (T_k)^c \ \ \textbf{;} \ \ Y^{'(j)}:j\in
(T_k)^c|X^{'(j)}:j\in (T_k)^c\big)\leq
$$$$I\big(X^{'(i)}:i\in T_k \ \ \textbf{;} \ \ Y^{'(j)}:j\in
(T_k)^c|X^{'(j)}:j\in (T_k)^c\big)$$ where in inequality
\ref{Eq:AppendixBEqA} we have used the fact that
$H(X^{'(i)}|Y^{(i)})$=0 to add $X^{'(j)}:\ j\in (T_k)^c$ in the
conditioning part of the mutual information term. We have also added
$X^{'(i)}:i\in T_k $, but this can not cause the expression
decrease. In the equations \ref{Eq:AppendixBEqB} and
\ref{Eq:AppendixBEqC} we have used the following Markov chain
$$\big(Y^{'(i)}:i\in [m]\big) -(X^{'(i)}:i\in [m]) -
(Y^{(i)}X^{(i)}:i\in [m]).$$

The $k^{th}$ coordinate of $\overrightarrow{v}$ is thus less than or
equal to $$I\big(X^{(i)}:i\in T_k \ \ \textbf{;} \ \ Y^{(j)}:j\in
(T_k)^c|X^{(j)}:j\in (T_k)^c\big)+$$$$I\big(X^{'(i)}:i\in T_k \ \
\textbf{;} \ \ Y^{'(j)}:j\in (T_k)^c|X^{'(j)}:j\in (T_k)^c\big).$$

Let $k^{th}$ coordinate of $\overrightarrow{v_1}$ be
$$I\big(X^{(i)}:i\in T_k \ \ \textbf{;} \ \ Y^{(j)}:j\in
(T_k)^c|X^{(j)}:j\in (T_k)^c\big),$$ and the $k^{th}$ coordinate of
$\overrightarrow{v_2}$ be $$I\big(X^{'(i)}:i\in T_k \ \ \textbf{;} \
\ Y^{'(j)}:j\in (T_k)^c|X^{'(j)}:j\in (T_k)^c\big).$$ \hfill
{$\blacksquare$}
\subsection{Checking the second property of Lemma \ref{Thm:MainThm} } \label{Appendix:ExtensionTwo}
Our choice of $\phi(.)$ implies
$$\phi\big(p(y^{(1)},...,y^{(m)}|x^{(1)},...,x^{(m)}),
\{q(x_1,...,x_m)\}\big)=\varphi
(p(y^{(1)},...,y^{(m)}|x^{(1)},...,x^{(m)})p(x^{(1)},...,x^{(m)})).$$
Take an arbitrary point $\overrightarrow{v}$ inside the above set.
The $k^{th}$ coordinate of $\overrightarrow{v}$ is less than or
equal to $I\big(X^{(i)}:i\in T_k \ \ \textbf{;} \ \ Y^{(j)}:j\in
(T_k)^c|X^{(j)}:j\in (T_k)^c\big)$ where $T_k$ is defined as in the
proof of Theorem \ref{Thm:CutSet}. Since $Y^{(j)}=X^{(j)}$ for $j\in
[m]$, the $k^{th}$ coordinate of $\overrightarrow{v}$ would be less
than or equal to zero. But $\overrightarrow{v}$ also lies in
$\mathbb{R}_{+}^c$, hence it has to be equal to the all zero vector.
\hfill {$\blacksquare$}

\subsection{Checking the third property of Lemma \ref{Thm:MainThm}} \label{Appendix:ExtensionThree}
Given the definition of $\phi(.)$ in equation
\ref{Eq:DefinitionOfHardPhi}, one needs to verify that:
\[\varphi\big(p(z^{(1)},...,z^{(m)}|x^{(1)},x^{(2)},...,x^{(m)})p(x^{(1)},x^{(2)},...,x^{(m)})\big)
\subseteq\]\[\varphi\big(p(y^{(1)},...,y^{(m)}|x^{(1)},x^{(2)},...,x^{(m)})p(x^{(1)},x^{(2)},...,x^{(m)})\big).\]
Take an arbitrary point $\overrightarrow{v}$ inside \[\varphi
(p(z^{(1)},...,z^{(m)}|x^{(1)},x^{(2)},...,x^{(m)})p(x^{(1)},...,x^{(m)})).\]
The $k^{th}$ coordinate of $\overrightarrow{v}$ is less than or
equal to $I\big(X^{(i)}:i\in T_k \ \ \textbf{;} \ \ Z^{(j)}:j\in
(T_k)^c|X^{(j)}:j\in (T_k)^c\big)$ where $T_k$ is defined as in the
proof of Theorem \ref{Thm:CutSet}. The latter vector itself is less
than or equal to a vector, denoted here by $\overrightarrow{v'}$,
whose $k^{th}$ coordinate is equal to $I\big(X^{(i)}:i\in T_k \ \
\textbf{;} \ \ Y^{(j)}:j\in (T_k)^c|X^{(j)}:j\in (T_k)^c\big)$
because $$p(z^{(1)},z^{(2)},...,z^{(m)},
y^{(1)},y^{(2)},...,y^{(m)}|x^{(1)},x^{(2)},...,x^{(m)})
    =$$$$p(y^{(1)},y^{(2)},...,y^{(m)}|x^{(1)},x^{(2)},...,x^{(m)})\prod_{i=1}^m p(z_i|y_i),$$
implying that for every $i \in [m]$,
$I(Z^{(i)}\textbf{;}G^{(i)}|Y^{(i)})$ is zero for $G^{(i)}$ defined
as follows:
$$G^{(i)}=(Z^{(1)},Z^{(2)},...,Z^{(i-1)},
Z^{(i+1)},..., Z^{(m)}, Y^{(1)},Y^{(2)},...,Y^{(i-1)},
Y^{(i+1)},..., Y^{(m)}, X^{(1)},X^{(2)},...,X^{(m)}).$$

Since the point $\overrightarrow{v'}$ is inside
\[\varphi
(p(y^{(1)},...,y^{(m)}|x^{(1)},x^{(2)},...,x^{(m)})p(x^{(1)},...,x^{(m)})),\]
we conclude that
\[\varphi\big(p(z^{(1)},...,z^{(m)}|x^{(1)},x^{(2)},...,x^{(m)})p(x^{(1)},x^{(2)},...,x^{(m)})\big)
\subseteq\]\[\varphi\big(p(y^{(1)},...,y^{(m)}|x^{(1)},x^{(2)},...,x^{(m)})p(x^{(1)},x^{(2)},...,x^{(m)})\big).\]
\hfill {$\blacksquare$}
\section{} \label{Appendix:ExtensionFive}
We will define random variables $\widetilde{\widehat{M}^{'(i)}}$
(for $i \in [m]$) such that for any $i \in [m]$
$$\mathbb{E}\big[\Delta_i(\widetilde{\widehat{M}^{'(i)}},
\widetilde{\widehat{M}^{(i)}})\big]\leq D^{(i)},$$ and furthermore
$$I\big(\widetilde{W}^{(i)}:i\in T_k \ \ \textbf{;} \ \
\widetilde{\widehat{M}^{'(j)}}:j\in (T_k)^c|\widetilde{W}^{(j)}:j\in
(T_k)^c\big)-O(\tau(\epsilon))\leq$$$$I\big(\widetilde{W}^{(i)}:i\in
T_k \ \ \textbf{;} \ \ \widetilde{\widehat{M}^{(j)}}:j\in
(T_k)^c|\widetilde{W}^{(j)}:j\in (T_k)^c\big),$$  where
$\tau(\epsilon)\rightarrow 0$ as $\epsilon \rightarrow 0$.
%

Intuitively speaking, the algorithm for creating
$\widetilde{\widehat{M}^{'(i)}}$ is to begin with
$\widetilde{\widehat{M}^{(i)}}$ ($i \in [m]$), and then perturbs
this set of $m$ random variables in $m$ stages as follows: at the
$r^{th}$ stage, we perturb the $r^{th}$ random variable so that the
average distortion constraint is satisfied while making sure that
changes in the mutual information terms are under control.

More precisely, let $(G_0^{(1)}, G_0^{(2)}, ... , G_0^{(m)})$ be
equal to $(\widetilde{\widehat{M}^{(1)}},
\widetilde{\widehat{M}^{(2)}},...,\widetilde{\widehat{M}^{(m)}})$.
 We define random variables $(G_r^{(1)}, G_r^{(2)}, ... , G_r^{(m)})$ for $\ r\in [m]$
 using $(G_{r-1}^{(1)}, G_{r-1}^{(2)}, ... , G_{r-1}^{(m)})$ in a
 sequential manner as follows: let $G_r^{(i)}:=G_{r-1}^{(i)}$ for all $i\in [m]$, $i\neq r$. Random
variable $G_r^{(r)}$ is defined below by perturbing $G_{r-1}^{(i)}$
in a way that the average distortion between $G_r^{(r)}$ and
$\widetilde{M}^{(r)}$ is less than or equal to $D^{(r)}$ while
making sure that for any $k \in [2^m-2]$,
$$I\big(\widetilde{W}^{(i)}:i\in T_k \ \ \textbf{;} \ \ G_r^{(j)}
:j\in (T_k)^c|\widetilde{W}^{(j)}:j\in
(T_k)^c\big)-I\big(\widetilde{W}^{(i)}:i\in T_k \ \ \textbf{;} \ \
G_{r-1}^{(j)} :j\in (T_k)^c|\widetilde{W}^{(j)}:j\in (T_k)^c\big)$$
is of order $O(\tau_r(\epsilon))$ where $\tau_r(.)$ is a real-valued
function that satisfies the property that
$\tau_r(\epsilon)\rightarrow 0$ as $\epsilon \rightarrow 0$. Once
this is done, we can take $\widetilde{\widehat{M}^{'(i)}}=G_m^{(i)}$
for all $i \in [m]$ and let
$\tau(\epsilon)=\sum_{r=1}^m\tau_r(\epsilon)$.

For any arbitrary $k \in [2^m-2]$, as long as $r$ does not belong to
$(T_k)^c$, the expression $$I\big(\widetilde{W}^{(i)}:i\in T_k \ \
\textbf{;} \ \ G_r^{(j)} :j\in (T_k)^c|\widetilde{W}^{(j)}:j\in
(T_k)^c\big)-$$$$I\big(\widetilde{W}^{(i)}:i\in T_k \ \ \textbf{;} \
\ G_{r-1}^{(j)} :j\in (T_k)^c|\widetilde{W}^{(j)}:j\in
(T_k)^c\big),$$ would be zero no matter how $G_r^{(r)}$ is defined.
We should therefore consider only the cases where $r$ belongs to
$(T_k)^c$. In order to define $G_r^{(r)}$, we consider two cases:
\begin{enumerate}
  \item Case $D^{(r)}\neq 0$:
Take a binary random variable $Q_r$ independent of all other random
variables defined in previous stages. Assume that
$P(Q_r=0)=\frac{\epsilon}{D^{(r)}+\epsilon}$ and
$P(Q_r=1)=\frac{D^{(r)}}{D^{(r)}+\epsilon}$. Let $G_r^{(r)}$  be
equal to $G_{r-1}^{(r)}$  if $Q_r=1$, and be equal to
$\widetilde{M}^{(r)}$ if $Q_r=0$. It can be verified that the
average distortion between $G_r^{(r)}$ and $\widetilde{M}^{(r)}$ is
less than or equal to $D^{(r)}$.\footnote{This is because
$\mathbb{E}\big[\Delta_r(G_{r}^{(r)}, \widetilde{M}^{(r)})\big]=
\mathbb{E}\big[\mathbb{E}\big[\Delta_r(G_{r}^{(r)},
\widetilde{M}^{(r)})|Q_r\big]\big]=P(Q_r=1)\mathbb{E}\big[\Delta_r(G_{r-1}^{(r)},
\widetilde{M}^{(r)})\big]\leq \frac{D^{(r)}}{D^{(r)}+\epsilon}\cdot
(D^{(r)}+\epsilon)=D^{(r)}.$}

Take an arbitrary $k \in [2^m-2]$ such that $r\in T_k$. Since for
any five random variables $A,B,B',C,D$ where $D$ is independent of
$(A,B,C)$ we have $I(A;B'|C)-I(A;B|C)\leq I(A;B'|BCD)$,
\footnote{This is because $I(A;B|C)\geq I(A;B'|C)-I(A;B'|BC)\geq
I(A;B'|C)-I(A;B'D|BC) \geq
I(A;B'|C)-I(A;D|BC)-I(A;B'|BCD)=I(A;B'|C)-0-I(A;B'|BCD)=I(A;B'|C)-I(A;B'|BCD).$}
we can write:
$$I\big(\widetilde{W}^{(i)}:i\in T_k \ \
\textbf{;} \ \ G_r^{(j)} :j\in (T_k)^c|\widetilde{W}^{(j)}:j\in
(T_k)^c\big)-$$$$I\big(\widetilde{W}^{(i)}:i\in T_k \ \ \textbf{;} \
\ G_{r-1}^{(j)} :j\in (T_k)^c|\widetilde{W}^{(j)}:j\in
(T_k)^c\big)\leq$$$$I\big(\widetilde{W}^{(i)}:i\in T_k \ \
\textbf{;} \ \ G_r^{(j)} :j\in
(T_k)^c|G_{r-1}^{(j)}\widetilde{W}^{(j)}:j\in (T_k)^c, Q_r\big).$$
We would like to prove that the last term is of order
$\tau_r(\epsilon):=O(\frac{\epsilon}{D^{(r)}+\epsilon})$. Clearly
then $\tau_r(\epsilon) \rightarrow 0$ as $\epsilon \rightarrow 0$
since $D^{(r)}$ is assumed to be non-zero. The last term above is of
order $\frac{\epsilon}{D^{(r)}+\epsilon}$
because:$$I\big(\widetilde{W}^{(i)}:i\in T_k \ \ \textbf{;} \ \
G_r^{(j)} :j\in (T_k)^c|G_{r-1}^{(j)}\widetilde{W}^{(j)}:j\in
(T_k)^c, Q_r\big)=$$$$0\cdot P(Q_r=1)+$$$$
I\big(\widetilde{W}^{(i)}:i\in T_k \ \ \textbf{;} \ \ G_r^{(j)}
:j\in (T_k)^c|G_{r-1}^{(j)}\widetilde{W}^{(j)}:j\in (T_k)^c,
Q_r=0\big)\cdot P(Q_r=0)\leq$$$$ H(\widetilde{W}^{(i)}:i\in
[m])\cdot P(Q_r=0)=O(\frac{\epsilon}{D^{(i)}+\epsilon}).$$
  \item Case $D^{(r)}= 0$: Let the binary random variable $Q_r$ be the indicator function
  $\textbf{1}[\Delta_r(G_{r-1}^{(r)}, \widetilde{M}^{(r)})=0]$.
  Let $G_r^{(r)}$  be equal to $G_{r-1}^{(r)}$  if $Q_r=1$, and be equal
to $\widetilde{M}^{(r)}$ if $Q_r=0$. The average distortion between
$G_r^{(r)}$ and $\widetilde{M}^{(r)}$ is clearly zero. Since the
average distortion between $G_{r-1}^{(r)}$ and $
\widetilde{M}^{(r)}$ is less than or equal to $\epsilon$, we get
that $P(Q_r=0)\leq \frac{\epsilon}{\delta_{min}}$ where
$\delta_{min}$ is defined as follows:
($\widetilde{\mathcal{M}}^{(r)}$ here refers to the set
$\widetilde{M}^{(r)}$ is taking value from)
$$\delta_{min}=\min_{\scriptsize{\begin{array}{l}
              i,j \in
\widetilde{\mathcal{M}}^{(r)} \mbox{ such that }\\ \Delta_r(i,
j)\neq 0
            \end{array}}}\Delta_r(i, j).$$

Take an arbitrary $k \in [2^m-2]$ such that $r\in T_k$.
$$I\big(\widetilde{W}^{(i)}:i\in T_k \ \
\textbf{;} \ \ G_r^{(j)} :j\in (T_k)^c|\widetilde{W}^{(j)}:j\in
(T_k)^c\big)-$$$$I\big(\widetilde{W}^{(i)}:i\in T_k \ \ \textbf{;} \
\ G_{r-1}^{(j)} :j\in (T_k)^c|\widetilde{W}^{(j)}:j\in
(T_k)^c\big)=$$
$$H\big(\widetilde{W}^{(i)}:i\in T_k|
G_r^{(j)}\widetilde{W}^{(j)}:j\in
(T_k)^c\big)-$$$$H\big(\widetilde{W}^{(i)}:i\in T_k|
G_{r-1}^{(j)}\widetilde{W}^{(j)}:j\in (T_k)^c\big)\leq$$
$$H(Q_r)+H\big(\widetilde{W}^{(i)}:i\in T_k|
G_r^{(j)}\widetilde{W}^{(j)}:j\in (T_k)^c,
Q_r\big)-$$$$H\big(\widetilde{W}^{(i)}:i\in T_k|
G_{r-1}^{(j)}\widetilde{W}^{(j)}:j\in (T_k)^c, Q_r\big)\leq$$
$$H(Q_r)+P(Q_r=0)\cdot H\big(\widetilde{W}^{(i)}:i\in T_k|
G_r^{(j)}\widetilde{W}^{(j)}:j\in (T_k)^c,
Q_r=0\big)\leq$$$$H(Q_r)+P(Q_r=0)\cdot H(\widetilde{W}^{(i)}:i\in
[m]).$$ Let $\tau_r(\epsilon):=H(Q_r)+P(Q_r=0)\cdot
H(\widetilde{W}^{(i)}:i\in [m])$. Since $P(Q_r=0)$ is bounded from
above by $\frac{\epsilon}{\delta_{min}}$ that converges to zero as
$\epsilon \rightarrow 0$, $\tau_r(\epsilon)$ too would converge to
zero as $\epsilon \rightarrow 0$.
\end{enumerate}
  \hfill
{$\blacksquare$}

\section*{Acknowledgement}
The authors would like to thank TRUST (The Team for Research in
Ubiquitous Secure Technology), which receives support from the
National Science Foundation (NSF award number CCF-0424422) and the
following organizations: Cisco, ESCHER, HP, IBM, Intel, Microsoft,
ORNL, Pirelli, Qualcomm, Sun, Symantec, Telecom Italia and United
Technologies, for their support of this work. The research was also
partially supported by NSF grants CCF-0500023, CCF-0635372, and
CNS-0627161.

\end{document}